\documentclass{article}
\usepackage{spconf}

\usepackage{appendix}
\usepackage{makecell}
\usepackage{enumitem}

\usepackage[english]{babel}

\usepackage{ifpdf}

\usepackage{cite} 
\usepackage{url}
\usepackage{hyperref}
\usepackage{enumitem}

\usepackage[pdftex]{graphicx}
\graphicspath{{./figures/}}
\usepackage{color}
\usepackage{pgf, tikz, pgfplots}
\usetikzlibrary{shapes, arrows, automata, plotmarks}
\usetikzlibrary{calc,hobby,decorations}
\usepackage{fixltx2e}
\usepackage{stfloats}

\usepackage[cmex10]{amsmath}
\usepackage{amsfonts, amssymb, amsthm}
\usepackage{mathrsfs}

\usepackage{algorithm,algpseudocode}
\algnewcommand{\LeftComment}[1]{\Statex \(\triangleright\) #1}

\usepackage{color, colortbl}
\definecolor{Gray}{gray}{0.9}
\usepackage{array}

\usepackage{enumitem}

\usepackage{multirow}
\usepackage{subcaption}

\input{pennColors.sty}
\usepackage{needspace}

\input{mySymbol.sty}

\def\wtbx{\widetilde{\bbx}}
\def\wtbL{\widetilde{\bbL}}

\def\h{\text{h}}
\def\c{\text{c}}
\def\g{\text{g}}

\newtheorem{proposition}{\hspace{0pt}\bf Proposition}

\newtheorem{remark}{\hspace{0pt}\bf Remark}

\newcommand{\QED}{\hfill\ensuremath{\blacksquare}}

\title{Convolutional Filtering in Simplicial Complexes}
\name{Elvin Isufi and Maosheng Yang\thanks{Faculty of Electrical Engineering, Mathematics and Computer Science, Delft University of Technology, Delft, The Netherlands. This work is supported by the TU Delft AI Labs Programme. e-mail: \{e.isufi-1\}@tudelft.nl}
\address{~}}
\begin{document}
\ninept
\maketitle
\begin{abstract}
This paper proposes convolutional filtering for data whose structure can be modeled by a simplicial complex (SC). SCs are mathematical tools that not only capture pairwise relationships as graphs but account also for higher-order network structures. These filters are built by following the shift-and-sum principle of the convolution operation and rely on the Hodge-Laplacians to shift the signal within the simplex. But since in SCs we have also inter-simplex coupling, we use the incidence matrices to transfer the signal in adjacent simplices and build a filter bank to jointly filter signals from different levels. We prove some interesting properties for the proposed filter bank, including permutation and orientation equivariance, a computational complexity that is linear in the SC dimension, and a spectral interpretation using the simplicial Fourier transform. We illustrate the proposed approach with numerical experiments.
\end{abstract}
\begin{keywords}
Hodge Laplacian, simplicial filter, topological signal processing. \vspace{-.5cm}\vskip-.5cm
\end{keywords}


\section{Introduction}\vskip-.15cm
\label{sec:intro}

Processing data with an irregular structure has been the center of the research attention in the last decade, generalizing signal processing \cite{ortega2018graph} and neural network techniques \cite{wu2020comprehensive} to graphs. One common factor behind the success of these two directions is the concept of graph filtering, that extends the convolution operation from the Euclidean domain to the graph domain \cite{sandryhaila13-dspg, shuman13-mag, gama2020graphs}. However, graph filters exploit only pair-wise relationships in the network and consider the data as signals over its vertices. But data often live on higher-order network structures such as edges and triangles \cite{lim2020hodge,barbarossa2020topological,schaub2021signal}. Typical examples include water flows in water networks or traffic flows in transportation networks.

To deal with this type of data and account for their structure, the more recent research attention has shifted towards data processing with a simplicial structure \cite{robinson2014topological, courtney2016generalized,barbarossa2020topologicala, schaub2021signal}, i.e., data living over edges, triangles, and so on. The work in \cite{schaub2018flow} considers the problem of flow denoising via simplicial regularization, while \cite{jia2019graph} focused on flow interpolation. Authors in \cite{barbarossa2020topological} introduced the concept of simplicial Fourier transform (SFT), which shows that the approaches in \cite{schaub2018flow, jia2019graph} behave as low-pass filtering. To further increase the filter flexibility, the work in \cite{yang2021finite} used the shift-and-sum principle to develop convolutional simplicial filtering via polynomials in the Hodge Laplacian \cite{lim2020hodge}. In parallel, \cite{ebli2020simplicial} proposed the simplicial convolutional networks by extending the popular graph neural networks (GNNs) to the simplex. The aggregation function in \cite{ebli2020simplicial} respects the convolution principle, and it is a particular case of the filter in \cite{yang2021finite}. Other simplicial neural networks include \cite{roddenberry2019hodgenet, roddenberry2021principled, bodnar2021weisfeiler}, which, although built on the principle of message passing \cite{gilmer2017neural} can be seen as an order one simplicial convolution nested into point-wise nonlinearities.

Despite the emerging success, simplicial convolutional filters operate only within a simplicial level and process only one type of simplicial signal. Consequently, they are agnostic to the full SC structure since simplicial signals have not only \emph{intra}-simplex proximities but also \emph{inter}-simplex proximities \cite{lim2020hodge}. E.g., a simplicial convolutional filter operating over the edge space processes only edge signals but ignores the effect of the adjacent vertex and triangle signals. For instance, in a water network we may have missing edge flow measurements but not nodal pressures; hence, we could use the latter and the SC coupling to infer the edge signals. To account for such a coupling in a principled way, we extend \cite{yang2021finite} to a filter bank to process jointly all SC signals.

Our specific contribution is threefold: i) we propose a filter bank composed of simplicial convolutional filters to process jointly signals living on different simplicial levels; ii) we show the filter bank operates by using local information and it is equivariant to permutations and orientations in the simplex; iii) we characterize the spectral response of the filter bank and show how the different simplicial signals are filtered only in different simplicial frequencies. \vspace{-.25cm}



\section{Simplicial Signal Processing}\vskip-.15cm
\label{sec:prelim}

In this section, we first lay down some basic concepts about simplicial complexes. Then, we discuss the simplicial signals.\vspace{-0.25cm}

\subsection{Geometry of Simplicial Complexes}\vspace{-.1cm}

Given a set of vertices $\ccalV = \{1, \ldots, N\}$, a $k-$simplex $\ccalS^k$ is a subset of $\ccalV$ containing $k+1$ distinct elements. 
%
%
Typical simplices are the $0$-simplex including a node, the $1-$simplex including an edge, and the $2-$simplex including a triangle \cite{lim2020hodge}. A simplicial complex of order $K$, $\ccalX^K$, is a collection of simplices such that for any $k-$simplex $\ccalS^k$, it includes any subset $\ccalS^{k-1}\subset \ccalS^k$ for all $k = 0, \ldots, K$. 
%
%
The number of $k-$simplices in $\ccalX^K$ is $N_k$. A simplicial complex of order $K = 2$ formed by two disjoint triangles $N_{\!2} = 2$, includes all their edges $N_1 \!=\! 6$ and nodes $N_{\!0} \!=\! 6$. We also say two simplices $\ccalS^{k}$ and $\ccalS^{k+1}\!$ are adjacent if they belong to the same SC.

We represent the proximities between the different simplices via the incidence matrices $\bbB_k \in \reals^{N_{k-1} \times N_k}$, which have as row index the $(k-1)-$simplices and as column index the $k-$simplices. For instance, $\bbB_1$ is the vertex-to-edge incidence matrix and $\bbB_2$ is the edge-to-triangle incidence matrix. \vspace{-.2cm} 

\begin{property}\label{prop_incl} Two adjacent incidence matrices $\bbB_k$ and  $\bbB_{k+1}$ satisfy the boundary condition $\bbB_k\bbB_{k+1} = \bbzero$ for all $k \ge 1$ \cite{lim2020hodge}.
\end{property}\vspace{-.2cm}

Using the incidence matrices, we can fully represent the structure of the SC by the Hodge Laplacian matrices, \vspace{-.2cm}
\begin{align}\label{eq.HodgeLapl}
\begin{split}
&\bbL_0 = \bbB_1\bbB_1^\top\\
&\bbL_k = \bbL_{k\ell} + \bbL_{ku} =\bbB_k^\top\bbB_k + \bbB_{k+1}\bbB_{k+1}^\top~k\!=\! 1, \!\ldots\!, K\!-\!1\\
&\bbL_K = \bbB_K^\top\bbB_K.
\end{split}
\end{align}
That is, the zero-Hodge Laplacian $\bbL_0 \!=\! \bbB_1\bbB_1^\top$ is the popular graph Laplacian and indicates vertex proximities based on their upper-adjacency via an edge \cite{shuman13-mag}. All the intermediate Laplacians $\bbL_k$ comprise two terms: the lower Laplacian $\bbL_{k\ell}:=\bbB_k^\top\bbB_k$ that captures lower-adjacencies of $k-$simplices; and the upper Laplacian $\bbL_{ku}\!:=\!\bbB_{k+1}\bbB_{k+1}$ that captures upper-adjacencies of $k-$simplices. E.g., in an SC of order $K = 2$ two edges are lower-adjacent if they have a common vertex and upper-adjacent if they belong to the same triangle. The $K-$th Hodge Laplacian $\bbL_K = \bbB_K^\top\bbB_K$ has only the lower-adjacencies since the SC is of order $K$. \vspace{-.35cm}

\subsection{Signals over Simplicial Complexes}\label{subsec_sigSC}\vspace{-.1cm}

We are interested in processing signals residing over the different simplicial levels by accounting for the overall structure of the simplicial complex. A $k-$simplicial signal, for short a $k-$signal, $x^k$ is a mapping from the $k-$simplex to the real set $\reals^{N_k}$, $x^k: \ccalS^k \to \reals^{N_k}$. We will collect the $k-$signal in vector $\bbx^k = [x^k_1, \ldots, x^k_{N_k}]^\top$, where $x_i^k$ is the signal on the $i$th simplex. The $0-$signal resides over the vertices and matches the graph signal \cite{shuman13-mag}. Likewise, the $1-$signal resides over the edges, the $2-$signal over the triangles, and so on. We will also refer to the collection of all $k-$signals $\bbx^0, \ldots, \bbx^K$ as a \emph{simplicial complex signal}.

Given the simplicial adjacencies within an SC [cf. \eqref{eq.HodgeLapl}] and the coupling between a simplicial complex and its signal, processing SC signals requires exploiting these interactions altogether. We do so by using the Hodge Laplacians in \eqref{eq.HodgeLapl} and their Hodge decomposition \cite{lim2020hodge}. The Hodge decomposition decomposes space $\reals^{N_k}$ of a $k-$signal into three orthogonal subspaces\vspace{-.05cm}
\begin{equation}\label{eq.spaceDecomp}
\reals^{N_k} = \text{im}(\bbB_k^\top) \oplus \text{im}(\bbB_{k+1}) \oplus \text{ker}(\bbL_k) 
\end{equation}
where $\text{im}(\cdot)$ and $\text{ker}(\cdot)$ are the image and kernels spaces of a matrix and $\oplus$ is the direct sum of vector spaces. This implies that for any $k-$signal $\bbx^k$ we can find three signals $\bbx^{k-1}$, $\bbx^k_\h$, and $\bbx^{k+1}$ of order $k-1$, $k$, and $k+1$, respectively such that we can write $\bbx^k$ as the sum of three orthogonal components\vspace{-.05cm}
\begin{equation}\label{eq.signDecomp}
\bbx^k = \bbB_k^\top\bbx^{k-1} + \bbB_{k+1}\bbx^{k+1} + \bbx_\h^k.
\end{equation}
This decomposition shows how the inter-simplex coupling imposed by the Hodge Laplacians in \eqref{eq.HodgeLapl} translates into an inter-signal coupling. Specifically:
\begin{itemize}
\item Operation $\bbx_\g^k :=  \bbB_k^\top\bbx^{k-1} \in \text{im}(\bbB_k^\top)$ transforms a $k-1$ signal $\bbx^{k-1}$ into the upper-simplex. Likewise, we can consider the opposite operation $\bbx^{k-1}_{\text{div}} = \bbB_k\bbx^{k}$ which transforms a $k-$signal into a signal in the lower $k-1$ simplex. Particularizing to $k = 1$, we have that $\bbx_\g^1 =  \bbB_1^\top\bbx^{0}$ is an edge flow induced by differentiating the vertex signals $\bbx^{0}$; and we refer to as the \emph{gradient flow}. Likewise, $\bbx^{0}_{\text{div}} := \bbB_1\bbx^{1}$ is a vertex signal obtained by computing the net flow of each node; and we refer to as the \emph{divergence component}.
\item Operation $\bbx_{\text{cur}} ^k := \bbB_{k+1}\bbx^{k+1} \in \text{im}(\bbB_{k+1})$ transforms a $k+1$ signal $\bbx^{k+1}$ into the lower-simplex. Likewise, through the adjoint $\bbB_{k+1}^\top$ we can transform a $k-$signal into the upper-simplex as $\bbx^{k+1}_\c:= \bbB_{k+1}^\top\bbx^k$. Particularizing again to $k = 1$, signal $\bbx_{\text{cur}}^1 = \bbB_{2}\bbx^{2}$ contains edge flows induced by triangle flows $\bbx^{2}$; and it is called a \emph{curl flow}. Likewise, signal $\bbx^{2}_\c = \bbB_{2}^\top\bbx^1$ is a flow circulating along the triangles computed from edge flows.
\item Signal $\bbx_\h^k \in \text{ker}(\bbL_k)$ is called the \emph{harmonic component} and it is that part of a $k-$signals that cannot be induced from the adjacent simplex signals. We can get the harmonic component by solving $\bbL_k\bbx_\h^k = \bbzero$.
\end{itemize}

Given this \emph{inter-simplex} coupling between different $k-$signals, we next leverage the shift-and-sum principle and the Hodge Laplacaians \eqref{eq.HodgeLapl} to induce an \emph{intra-simplex} coupling and develop a principled convolutional filter for SC signals. \vspace{-.2cm}


\section{Filters on Simplicial Complexes}
\label{sec:filt}


\subsection{Simplicial Convolutional Filters}\label{subsec_scf}

For a $k-$signal $\bbx^k$ over a simplex $\ccalS^k$ with Hodge Laplacian $\bbL_k$, $k = 1, \ldots, K-1$, a \emph{simplicial convolution} is defined as
\begin{equation}\label{eq.SimpConc}
\bby^k = \bigg(\sum_{l_1 = 0}^{L_1}\alpha_{l_1}\bbL_{k\ell}^{l_1} + \sum_{l_2 = 0}^{L_2}\beta_{l_2}\bbL_{ku}^{l_2}			\bigg)\bbx^k
\end{equation}
where $\alpha_{0}, \ldots, \alpha_{L_1}$ and $\beta_{0}, \ldots, \beta_{L_2}$ are parameters and $L_1$, $L_2$ are the convolutional orders in the lower-Laplacian $\bbL_{k\ell}$ and upper-Laplacian $\bbL_{ku}$, respectively. Defining then the \emph{simplicial convolutional filtering} matrix
\begin{equation}\label{eq.simplFilter}
 \bbH(\bbL_k):= \bigg(\sum_{l_1 = 0}^{L_1}\alpha_{l_1}\bbL_{k\ell}^{l_1} + \sum_{l_2 = 0}^{L_2}\beta_{l_2}\bbL_{ku}^{l_2}			\bigg)
\end{equation}
we can write \eqref{eq.SimpConc} as $\bby^k =  \bbH(\bbL_k)\bbx^k$. The qualifier convolution comes from the fact that in \eqref{eq.SimpConc} we are shifting the input $\bbx^k$ over simplex $\ccalS^k$ using both its lower- and upper-adjacencies, weighting each shift, and summing all shifted versions. This is analogous to the convolutional operator in discrete-time signal processing \cite{oppenheim2001discrete} and in graph signal processing \cite{sandryhaila13-dspg}. Filter $ \bbH(\bbL_k)$ is the sum of two polynomials because of the form of the Hodge Laplacian $\bbL_k$ [cf. \eqref{eq.HodgeLapl}] and of Property~\ref{prop_incl} (i.e., cross-terms are not present since $\bbL_{k\ell}^{l_1}\bbL_{ku}^{l_2} = \bbzero$). For $k = 0$ the simplicial convolutional filter reduces to the graph filter $\bbH(\bbL_0) = \sum_{l = 0}^{L}\alpha_{l}\bbL_{0\ell}^{l}$ \cite{sandryhaila13-dspg} and for $k = K$ to $\bbH(\bbL_K) = \sum_{l = 0}^{L}\beta_{l}\bbL_{K\ell}^{l}$ because $\bbL_K$ comprises only lower-adjacencies. 

It follows from \eqref{eq.SimpConc} that a simplicial convolutional filter acts on the $k-$signal and propagates neighboring information within simplex $\ccalS^k$ by leveraging paths either via lower- or upper-adjacencies. Computing the output in \eqref{eq.SimpConc} implies accounting for a local information from at most $L = \max\{L_1, L_2\}$ hops away following either of these paths \cite{yang2021finite}. Exploiting the recursions $\bbL_{k\ell}^{l_1}\bbx^k = \bbL_{k\ell}(\bbL_{k\ell}^{l_1-1}\bbx^k)$ and  $\bbL_{ku}^{l_2}\bbx^k = \bbL_{ku}(\bbL_{ku}^{l_2-1}\bbx^k)$ the cost of running a simplicial convolution is of order $\ccalO(LN_k)$. 

But filter \eqref{eq.simplFilter} ignores any influence of signals in the adjacent simplices. Since different signals $\bbx^0, \ldots, \bbx^K$ influence each other via the SC localities, we extend filter \eqref{eq.simplFilter} to account for the latter. For instance, in a water network we may want to process jointly measurements over the junctions (nodes), flows over pipes (edges), and catchments (triangles) to infer an output related to edge flows \cite{herrera2016graph}. 

\subsection{Simplicial Complex Filter Bank}\label{subsec_scfb_prop}

Given an SC signal $\bbx^0, \ldots, \bbx^K$, we define an SC filter bank as that generating the output SC signal 
\begin{align}\label{eq.simpCompFilt}
\begin{split}
\bby^0 &=  \bbH_0(\bbL_0)\bbx^0 +  \bbH_1(\bbL_0)\bbB_1\bbx^1\\
\bby^k &= \bbH_0(\bbL_k)\bbB_k^\top\bbx^{k-1} + \bbH_1(\bbL_k)\bbx^k + \bbH_2(\bbL_k)\bbB_{k+1}\bbx^{k+1}\\
\bby^K &= \bbH_0(\bbL_K)\bbB_K^\top\bbx^{K-1} + \bbH_1(\bbL_K)\bbx^K
\end{split}
\end{align}
where filters $\bbH_i(\bbL_k)$ are of the form in \eqref{eq.simplFilter} and subscript $i$ indicates they have different coefficients. 

This SC filter bank is local in two directions. First, it is \emph{intra}-simplex local since filters $\bbH_i(\bbL_k)$ capture information up to $L-$hops away within the $k-$simplex. Second, it is \emph{inter}-simplex local since computing output $\bby^k$ implies accounting also for the adjacent $(k-1)-$signal $\bbx^{k-1}$ and  $(k+1)-$signal $\bbx^{k+1}$. The $0-$output signal is computed by (graph) filtering with $\bbH_0(\bbL_0)$ the vertex signal $\bbx^{0}$ and by (graph) filtering with $\bbH_1(\bbL_0)$ the divergence signal $\bbx^{0}_{\text{div}} = \bbB_1\bbx^{1}$. Likewise, we can write the input-output relation for the $1-$output signal as
\begin{align}
\begin{split}
\bby^1 
= \bbH_0(\bbL_1)\bbx_\g^1 + \bbH_1(\bbL_1)\bbx^1 + \bbH_2(\bbL_1)\bbx_{\text{cur}}^1
\end{split}
\end{align}
which is now simplicial filtering the gradient flow (induced by node signals) $\bbx_\g^1 = \bbB_1^\top\bbx^{0}$, the input edge flow signals $\bbx^1$, and the curl flow induced by triangle signals $\bbx_{\text{cur}}^1:=\bbB_{2}\bbx^{2}$. The same discussion extends to any simplex $\ccalS^k$. This joint locality explores the computational benefits of the simplicial convolutional filter [cf. \eqref{eq.simplFilter}] and getting the output SC signal $\bby^0, \ldots, \bby^K$ with a computational cost of order $\ccalO(KLN)$, where $N = \max\{N_0, \ldots, N_K\}$ is the maximum number simplices. The linear complexity in $N$ makes the filter bank in \eqref{eq.simpCompFilt} practical even for SCs of large dimension, which links well with the linear-complexity of graph filters \cite{segarra17-linear,coutino2019advances}.

\smallskip
\noindent\textbf{Invariances.} As we process graph signals with an arbitrary node labeling, also in SCs we have simplicial signals with arbitrary labeling of the nodes and with an arbitrary orientation of the flows \cite{schaub2021signal}. Thus, studying the invariances of the filer bank \eqref{eq.simpCompFilt} shows how it exploits the symmetries in the simplicial complexes. 


\begin{proposition}[Permutation equivariance]\label{prop_permEquiv} Let $\ccalX^K$ be a simplicial complex and consider the permutation matrices $\bbP_k$ as those belonging to the set
\begin{equation*}
\ccalP = \{\bbP_k \in \{0, 1\}^{N_k \times N_k} : \bbP_k \bbone = \bbone, \bbP_k^\top\bbone = \bbone, k \ge 0		\}.
\end{equation*}
Permutation matrices $\bbP_k$ are such that products ${\wtbx}^k = \bbP_k^\top\bbx^k$ are reorderings of the entries of $\bbx^k$ and that the permuted Hodge Laplacian ${\wtbL}_k = \bbP_k^\top\bbL_k\bbP_k$ is a reordering of the rows and columns of $\bbL_k$. The simplicial complex filter bank in \eqref{eq.simpCompFilt} is permutation equivariant.
\end{proposition}

\begin{proof} \emph{(Sketch)} Consider the $k$th input-output relation in \eqref{eq.simpCompFilt}. The permutation matrices $\bbP_k$ transform the incidence matrices as $\widetilde{\bbB}_k = \bbP_{k-1}^\top\bbB_k\bbP_k$ and $\widetilde{\bbB}_{k+1} = \bbP_k^\top\bbB_{k+1}\bbP_{k+1}$. Then, using the identify $\bbP_k^\top\bbP_k = \bbP_k\bbP_k^\top = \bbI$, the permuted lower- and upper-Laplacians are respectively $\widetilde{\bbL}_{k\ell} = \bbP_k^\top\bbL_{k\ell}\bbP_k$ and $\widetilde{\bbL}_{ku} = \bbP_{k+1}^\top\bbL_{ku}\bbP_{k+1}$. Consequently, we can write $\bbH(\widetilde{\bbL}_k) = \bbP_k^\top\bbH({\bbL}_k)\bbP_k$. Using the above, we can apply the permuted filters to the permuted inputs $\widetilde{\bbx}^{k-1} = \bbP_{k-1}^\top\bbx^{k-1}$, ${\wtbx}^k = \bbP_k^\top\bbx^k$, and ${\wtbx}^{k+1} = \bbP_{k+1}^\top\bbx^{k+1}$ and prove that the output is permuted likewise, i.e., $\widetilde{\bby}^k = \bbP_k^\top\bby^k$. 
\end{proof}



\begin{proposition}[Orientation equivariance]\label{prop_orientEquiv}
Let $\ccalX^K$ be a simplicial complex of order $K$ and consider the orientation matrices $\bbD_k$ as those belonging to the set
\begin{equation*}
\ccalD = \{\bbD_k = \diag(\bbd_k): \bbd_k \in \{\pm 1\}^{N_k}, k \ge 1, \bbd_0 = \bbone	\}.
\end{equation*}
Orientation matrices $\bbD_k$ are such that products ${\wtbx}^k = \bbD_k^\top\bbx^k$ are reorientations of the flow directions in vectors $\bbx^k$ and that the reoriented Hodge Laplacian is given by ${\wtbL}_k = \bbD_k\bbL_k\bbD_k$. The simplicial complex filter bank in \eqref{eq.simpCompFilt} is orientation equivariant.
\end{proposition}

\begin{proof} \emph{(Sketch)} Consider the $k$th input-output relation in \eqref{eq.simpCompFilt}. The orientation matrices $\bbD_k$ transform the incidence matrices as $\widetilde{\bbB}_k = \bbD_{k-1}\bbB_k\bbD_k$ and $\widetilde{\bbB}_{k+1} = \bbD_k\bbB_{k+1}\bbD_{k+1}$. Using then the identity $\bbD_k\bbD_k = \bbI$, the oriented lower- and upper-Laplacians are respectively $\widetilde{\bbL}_{k\ell} = \bbD_k\bbL_{k\ell}\bbD_k$ and $\widetilde{\bbL}_{ku} = \bbD_{k+1}\bbL_{ku}\bbD_{k+1}$, which in turn leads to the oriented filters $\bbH(\widetilde{\bbL}_k) = \bbD_k\bbH({\bbL}_k)\bbD_k$. Using direct substitutions and simple algebra we can prove $\widetilde{\bby}^k = \bbD_k\bby^k$.
\end{proof}

Proposition~\ref{prop_permEquiv} (resp.~\ref{prop_orientEquiv}) implies that if relabel the SC (resp. reorient the flows) and apply the filter bank \eqref{eq.simpCompFilt}, the output is a relabeled (resp. reoriented) version of the output we would have gotten by applying the filter bank before relabelling (resp. reorientation). These equivariances also imply that we can learn the filter bank to process a given simplicial complex by seeing as examples only permuted and reoriented versions of it. I.e., if two parts of the SC are topologically identical and the simplices support identical signal flows, an SC filter bank yields identical outputs. These findings generalize the permutation equivariance seen for graph filters \cite{gama2020graphs,Isufi21-EdgeNets}.

\begin{remark}
\label{rem_simplAdj} 
When operating on signals transformed from lower-/upper-adjacent simplices, the simplicial filter bank uses only lower/upper- Laplacians. Consider the expression for $\bby^k$ for $k = 1, \ldots, K-1$ in \eqref{eq.simpCompFilt}. The first term reduces to $\bbH_0(\bbL_k)\bbB_k^\top\bbx^{k-1} = \bbH_0(\bbL_{k\ell})\bbB_k^\top\bbx^{k-1} + \beta_{00} \bbB_k^\top\bbx^{k-1}$ because of Property~\ref{prop_incl}. 
 Likewise, the third term reduces to $\bbH_2(\bbL_k)\bbB_{k+1}\bbx^{k+1} = \alpha_{20}\bbB_{k+1}\bbx^{k+1} + \bbH_2(\bbL_{ku})\bbB_{k+1}\bbx^{k+1} $.
 That is, a signal coming from lower simplices does not propagate into upper-adjacency paths and viceversa.
\qed
\end{remark}


\section{Frequency Response}
\label{sec:spec_resp}

We now analyze the properties of the SC filter bank in the spectral domain to give further insight into its filtering behavior.

\subsection{Simplicial Fourier Transform}

The $k-$th Hodge Laplacian can be eigendecomposed as $\bbL_k = \bbU_k \bbLambda_k\bbU_k^\top$, where matrix $\bbU_k = [\bbu_{k1}, \ldots, \bbu_{kN_k}]$ collects the eigenvectors and $\bbLambda_k = \diag(\lambda_{k1}, \ldots, \lambda_{kN_k})$ the eigenvalues on the main diagonal. For a $k-$signal $\bbx^k$, the simplicial Fourier transform (SFT) is defined as $\hbx^k = \bbU_k^\top\bbx^k$ \cite{barbarossa2020topological}. The inverse SFT is $\bbx^k = \bbU\hbx^k$. We refer to the eigenvalues $\lambda_{ki}$ as the simplicial frequencies. The SFT generalizes the graph Fourier transform \cite{barbarossa2020topological}.

The eigenvectors of $\bbL_k$ span the three subspaces of the Hodge decomposition [cf. \eqref{eq.spaceDecomp}]. That is, there exists some orthogonal eigenvectors $\bbU_{\g k} \in \reals^{N_k \times N_{g}}$ that span $\text{im}(\bbB_k^\top)$, $\bbU_{\c k} \in \reals^{N_k \times N_{c}}$ that span $\text{im}(\bbB_{k+1})$, and $\bbU_{\h k} \in \reals^{N_k \times N_{h}}$ that span $\text{ker}(\bbL_k)$. We collect the corresponding eigenvalues (simplicial frequencies) in sets $\ccalQ_\g = \{\lambda_{\g 1}, \ldots, \lambda_{\g N_\g}\}$, $\ccalQ_\c = \{\lambda_{\c 1}, \ldots, \lambda_{\c N_\c}\}$, and $\ccalQ_\h = \{\lambda_{\h 1} = 0, \ldots, \lambda_{\h N_\h} = 0\}$. Using these eigenvectors, we can project $k-$signals onto the respective spectral components as: $\hbx^k_\g = \bbU_{\g k}^\top\bbx^k$, $\hbx^k_\c = \bbU_{\c k}^\top\bbx^k$, and $\hbx^k_\h = \bbU_{\h k}^\top\bbx^k$, which show how the simplicial Fourier coefficients are spread among the three types of simplicial frequencies. Next, we shall see how the simplicial filter bank acts on these projections and achieves the desired filtering.

\subsection{Frequency Response of SC Filters}

To understand the spectral behavior of filter bank \eqref{eq.simpCompFilt}, we need first to understand the spectral behavior of the simplicial convolutional filter \eqref{eq.simplFilter}. Using the above discussion and the eigendecomposition of the $k$th Hodge Laplacian, the frequency response of the simplicial convolutional filter \eqref{eq.simplFilter} is
\begin{align}\label{eq.scf_resp}
\begin{split}
   \widehat{H}(\lambda_{ki})=\left\{
                \begin{array}{ll}
                  \alpha_0 + \beta_0 & \for~\lambda_{ki} \in \ccalQ_\h\\
                  \beta_0 + \sum_{l_1 = 0}^{L_1}\alpha_{l_1}\lambda_i^{l_1} & \for~\lambda_{ki} \in \ccalQ_\g\\
                  \alpha_0 + \sum_{l_2 = 0}^{L_2}\beta_{l_2}\lambda_i^{l_2} & \for~\lambda_{ki} \in \ccalQ_\c
                \end{array}
              \right.
\end{split}.
\end{align}
That is, we have independent control on the simplicial frequencies $\ccalQ_\g$ through $\alpha_{0}, \ldots, \alpha_{L_1}$ and on $\ccalQ_\c$ through $\beta_{0}, \ldots, \beta_{L_2}$ but we have no independent control on the simplicial frequencies $\ccalQ_\h$ \cite{yang2021finite}. 
The filter input-output relationship at the $i$th frequency is $\hat{y}_i^k = \widehat{H}(\lambda_i)\hat{x}_i^k$; i.e., it respects the convolution theorem by operating a point-wise multiplication in the SFT domain.

\begin{figure}[t!]
  \vspace{0mm}
  \centering
  \includegraphics[width=1\linewidth,scale=1]{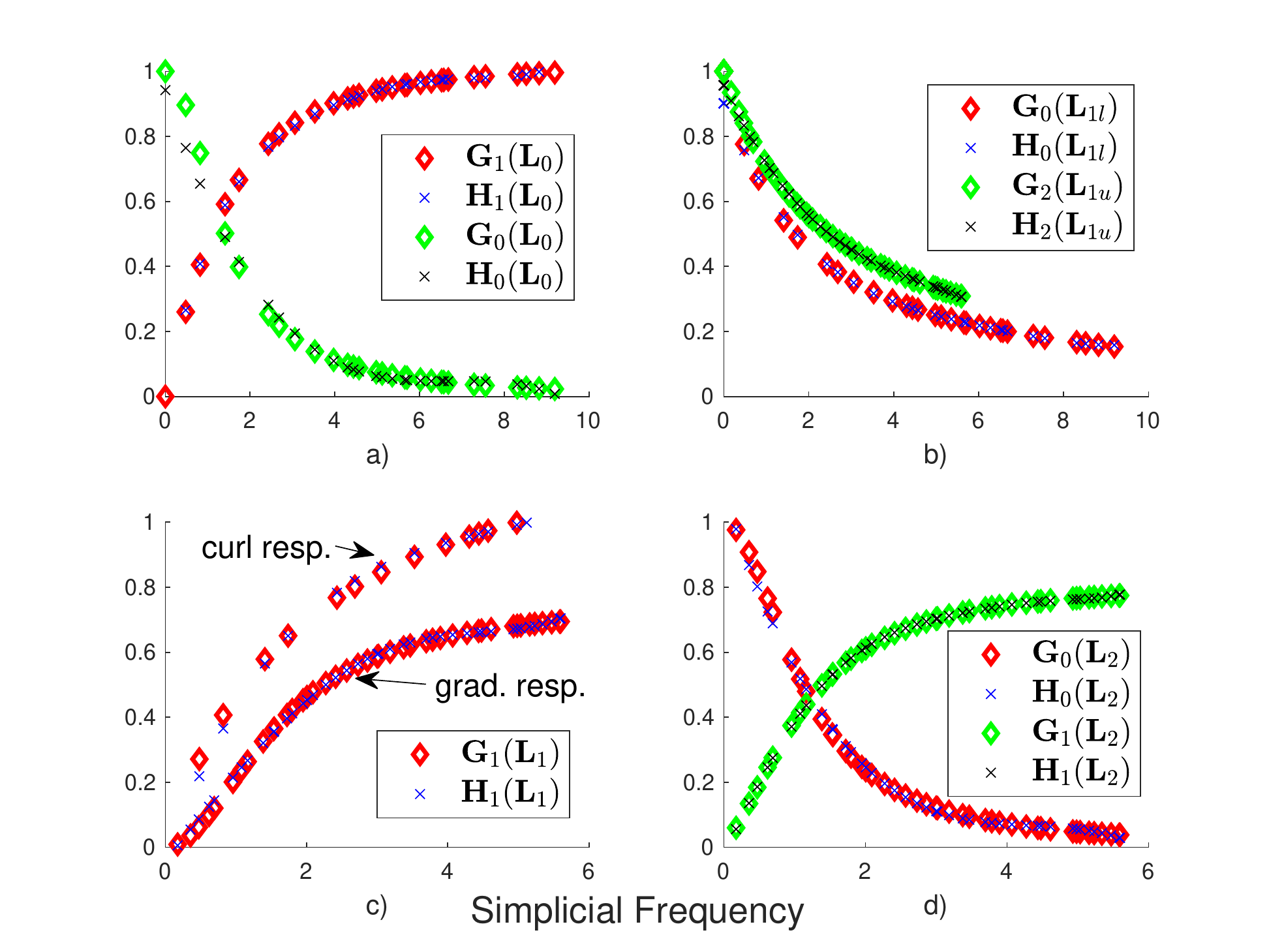}
  \vspace{-6mm}
  \caption{Data-driven approximation of inverse filtering. The $\diamond$ markers are the inverse filters, the $\times$ markers are their approximation with the filter bank \eqref{eq.simpCompFilt}. a) node filters $k = 1$; b) edge filters $k = 1$ for the gradient and curl influences; c) edge filters $k = 1$ for the curl and gradient frequencies; d) triangle filters $k = 2$.}
  \vspace{0mm}
  \label{fig:simplicia_filter_bank_frequency_response}
\end{figure}

Using these insights, let us now analyze the frequency behavior of the $k$th input-output relation in \eqref{eq.simpCompFilt}. From Remark~\ref{rem_simplAdj}, we can write the $k$th branch as
\begin{align}\label{eq.kbranchFB}
\begin{split}
\bby^k \!\!\!=\!\! \bbH_0(\bbL_{k\ell}\!)\bbx^k_\g \!\!+\! \beta_{00}\bbx^k_\g  \!+\! \bbH_1(\bbL_k\!)\bbx^k \!\!+\!\alpha_{20}\bbx_{\text{cur}}^k \!+\! \bbH_2(\bbL_{ku}\!)\bbx_{\text{cur}}^k
\end{split}
\end{align}
Using the SFT, its relation with the Hodge decomposition [cf. \eqref{eq.spaceDecomp}], and the filter response in \eqref{eq.scf_resp}, the $i$th input-output spectral relation is
%

\begin{align}
  \begin{split}
     \hat{y}_i^k \!\!=\!\!\left\{\!\!\!\!
                  \begin{array}{ll}
                     \widehat{H}_1(\lambda_{ki})\hat{x}^k_i &  \!\!\!\!\for~\lambda_{ki} \in \ccalQ_\h\\
                    \widehat{H}_0(\lambda_{ki})\hat{x}^k_{\g i} + \widehat{H}_1(\lambda_{ki})\hat{x}^k_i &\!\!\!\! \for~\lambda_{ki} \in \ccalQ_\g\\
                    \widehat{H}_1(\lambda_{ki})\hat{x}^k_i + \widehat{H}_2(\lambda_{ki})\hat{x}_{\text{cur} i}^k &\!\!\!\! \for~\lambda_{ki} \in \ccalQ_\c
                  \end{array}
                \right.
  \end{split}
  \end{align}
%
%
%
%
where $\hat{x}^k_{\g i}$ is the $i$th SFT coefficient of $\bbx_\g^k = \bbB_k^\top\bbx^{k-1}$ and $\hat{x}_{\text{cur} i}^k$ is the $i$th SFT coefficient of $\bbx_{\text{cur}}^k = \bbB_{k+1}\bbx^{k+1}$. Note that $\hat{H}_0$, $\hat{H}_1$ and $\hat{H}_2$ take the form in \eqref{eq.scf_resp}. This implies that filtering the signal from lower-adjacent simplices $ \bbB_k^\top\bbx^{k-1}$ does not play a role in $\ccalQ_\h$ and $\ccalQ_\c$ as $\bbx_\g^k$ contains only gradient spectral components; and likewise filtering the signal from upper-adjacent simplices $ \bbB_k^\top\bbx^{k-1}$ does not play a role over $\ccalQ_h$ and $\ccalQ_\g$ as $\bbx_\text{cur}^k$ contains only curl spectral components. This is intuitive because, for example, for $k=1$, edge flows induced from a node signal have no harmonic and curl components (being curl-free), and the edge flows induced from a triangle signal have no harmonic and gradient components (being divergence-free).


\begin{figure}[t!]
  \centering
  \includegraphics[width=1\linewidth,scale=1]{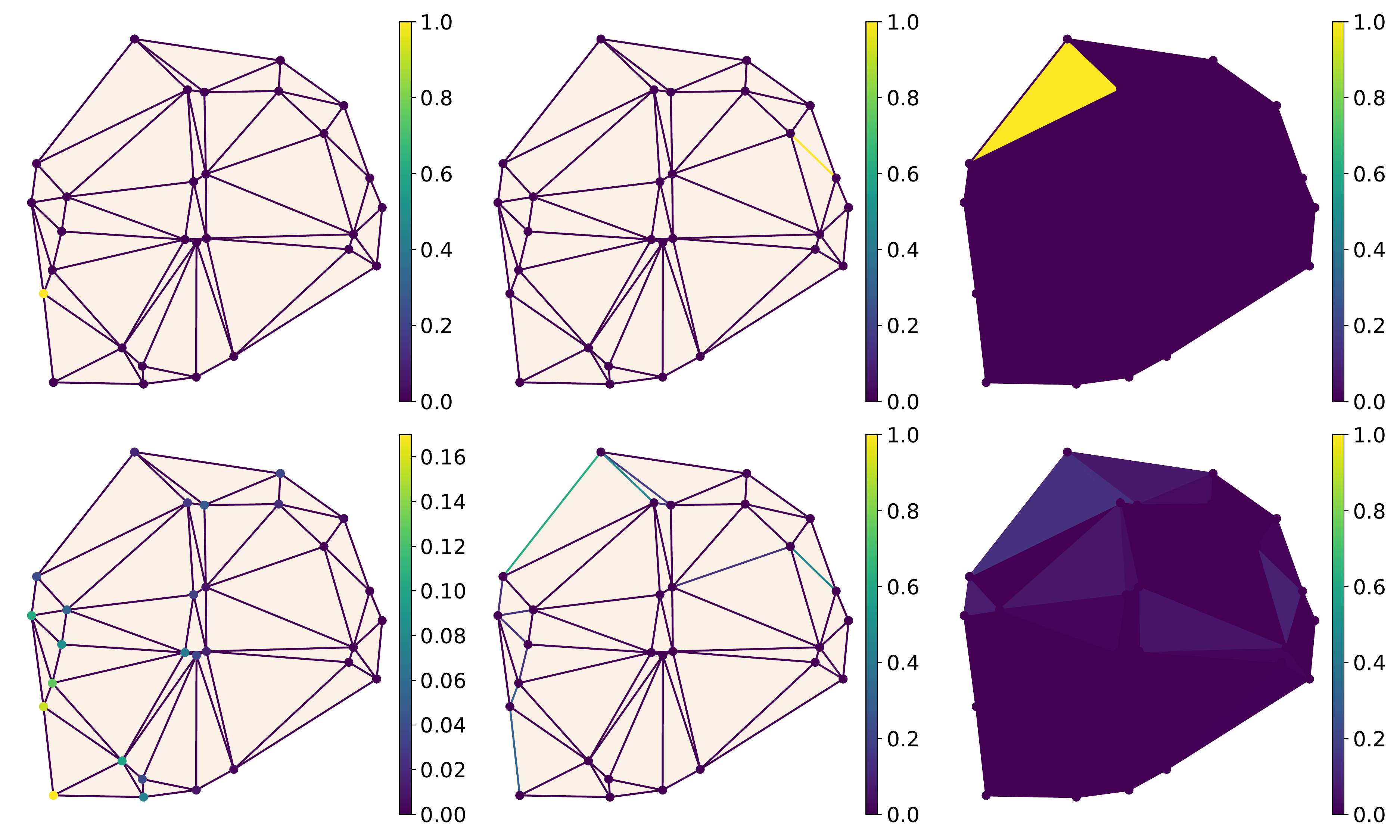}
  \vspace{-6mm}
  \caption{GHD based on the kernel approximation by SC filter banks. Top three (from left to right): node, edge, and triangle indicator input signals, respectively. Bottom three (from left to right): the diffused results on nodes, edges and triangles.}
  \label{fig:heat-diffusion}
\end{figure}

\section{Numerical Results}\vspace{-.2cm}
\label{sec:num_res}

We test the proposed filter bank for two tasks. First, we use it to fit inverse SC filtering in a data-driven approach. Second, we implement a heat kernel diffusion on SC with a low complexity. For both experiments, we generated an alpha SC of 29 nodes, 71 edges and 43 triangles with Gudhi toolbox \cite{gudhi:urm,gudhi:AlphaComplex}.


\smallskip
\noindent \textbf{Model fitting.} We generated 10 training samples $(\bbx^k,\bby^k)$ for $k=0,1,2$ by inputing a random simplicial singal $\bbx^k$ drawn from a zero-mean normal distribution, and outputting $\bby^k$ based on model \eqref{eq.simpCompFilt}. We considered rational filters $\bbG_i(\bbL_k) = \bbH_m^{-1}(\bbL_k) \bbH_n(\bbL_k)$ with $m$ and $n$ indicate different filter parameters. By stacking the shifted input and output training samples, we can solve least-squares problems to design an SC filter bank \eqref{eq.simpCompFilt} to fit the desired model parametrized by $\bbG_i(\bbL_k)$. With filter orders between 4 and 7, we achieved NMSEs of $0.03,0.01$ and $0.02$ for $k=0,1,2$, respectively, while with the simple simplicial filter \cite{yang2021finite} we get errors of at least one order higher (resp. 0.3, 0.68, and 0.9). Fig. 1 further shows how the different rational filter responses are well-approximated by the filter bank.



\smallskip
\noindent \textbf{Generalized heat diffusion (GHD).} The GHD is used to smooth meshes and identify key signatures in them \cite{zobel2011generalized}. The GHD behaves as an SC filter of the form $\bbG(\bbL_k) = \exp(-\gamma_k\bbL_k^2)$ which is computationaly heavy to compute because of the exponential matrix. Instead, we use its analytic frequency response $\widehat{G}(\lambda) = \exp(-\gamma_k\lambda^2_k)$ and universally approximate it with the convolutional filters \cite{yang2021finite}, with an approximation error smaller than $0.1$ but with an implementation cost of around two orders lower. Then, we used such filters for the filter bank in \eqref{eq.simpCompFilt} with $\gamma_0=0.3, \gamma_1=0.05$, and $\gamma_2=0.5$ to see how indicator input signals $\bbx^k$ diffuse within their simplices and in the neighboring ones, illustrated in Fig. \ref{fig:heat-diffusion}. We also observe that if $\gamma_k$ is large, the diffusion on the $k$-simplices attenuates faster, which is expected from the frequency response $\hat{G}(\lambda)$.

\section{Conclusion}
\label{sec:concl}

We proposed a simplicial complex convolutional filter bank that can process signals defined on different levels of the SC by capturing both their intra- and inter-simplex proximities. The intra-simplex proximities are captured by leveraging the shift-and-sum principle of the convolutional operation via the Hodge Laplacian matrices of a simplicial complex. Instead, the inter-simplex proximities are captured by leveraging the incidence matrices to transform the signal onto adjacent simplices and then filter with a simplicial filter defined on the adjacent simplices. We show the proposed filter bank is local and equivariant to both permutations in the simplex labeling and flow orientation. We also analyze the filter bank in the simplicial spectral domain and show it acts as a point-wise multiplication between the filter's frequency response and the simplicial Fourier transform of the signal, respecting the convolution principle.

\newpage
\bibliographystyle{IEEEbib}
\bibliography{myIEEEabrv,bib-nonlinear}

\end{document}